\documentclass[a4paper,11pt]{article}

\usepackage[top=2.7cm, bottom=3cm, left=2.5cm, right=2.5cm,
            footskip=.7in]{geometry}

\usepackage{amsmath}
\usepackage{amssymb}
\usepackage{amsfonts,amsthm,complexity}
\usepackage{thmtools,thm-restate}
\usepackage{tabularx,lipsum,environ}
\usepackage{multirow}
\usepackage{multicol}
\usepackage{hyperref}
\usepackage{xspace}
\usepackage{xifthen}
\usepackage{todonotes}

\usepackage{caption}
\usepackage{subcaption}

\usepackage{authblk}

\usepackage[T1]{fontenc}

\usepackage{lmodern}
\usepackage{setspace}

\newtheorem{theorem}{Theorem}
\newtheorem{lemma}{Lemma}
\newtheorem{Rule}{Rule}[theorem]

\newtheorem{claim}{Claim}[theorem]
\newtheorem{corollary}{Corollary}

\newtheorem{observation}{Observation}

\newcommand{\Oh}{\mathcal{O}}
\newcommand\yes{\textsc{Yes}}
\newcommand\no{\textsc{No}}

\newenvironment{claimproof}[1]{\par\noindent\emph{Proof:}\space#1}{{\leavevmode\unskip\penalty9999 \hbox{}\nobreak\hfill\quad\hbox{$\lrcorner$}}\medskip}

\usepackage{tikz}
 \usetikzlibrary{calc}
\usepackage{thmtools,thm-restate}

\newcommand\cliqueroot{\textsc{Distance-$k$-to-Clique Square Root}}
\newcommand\dpqroot{\textsc{Distance-$k$-to-$(pK_1+qK_2)$ Square Root}}

\newcommand\sroot{\textsc{VC-$k$ Root}}
\newcommand\bccover{\textsc{Biclique Cover}}

\usepackage{enumitem}
\setlist[itemize]{noitemsep}
\setlist[enumerate]{noitemsep}

\newcommand\ematch{\sim_{\text{mt}}}
\newcommand\enest{\sim_{\text{nt}}}
\newcommand\nematch{\nsim_{\text{mt}}}
\newcommand\nenest{\nsim_{\text{nt}}}

\newcommand{\ignore}[1]{}

\usepackage{framed}

\newlength{\RoundedBoxWidth}
\newsavebox{\GrayRoundedBox}
\newenvironment{GrayBox}[1]%
   {\setlength{\RoundedBoxWidth}{.93\textwidth}
    \def\boxheading{#1}
    \begin{lrbox}{\GrayRoundedBox}
       \begin{minipage}{\RoundedBoxWidth}}%
   {   \end{minipage}
    \end{lrbox}
    \begin{center}
    \begin{tikzpicture}%
       \node(Text)[draw=black!20,fill=white,rounded corners,%
             inner ysep=0.5ex,inner xsep=2ex,text width=\RoundedBoxWidth]%
             {\usebox{\GrayRoundedBox}};
        \coordinate(x) at (current bounding box.north west);
        \node [draw=white,rectangle,inner sep=3pt,anchor=north west,fill=white]
        at ($(x)+(6pt,.75em)$) {\boxheading};
    \end{tikzpicture}
    \end{center}}

\newenvironment{defproblemx}[2][]{\noindent\ignorespaces%
                                \FrameSep=6pt%
                                \parindent=0pt%
                \vspace*{-1.5em}
                \ifthenelse{\isempty{#1}}{%
                  \begin{GrayBox}{\textsc{#2}}%
                }{%
                  \begin{GrayBox}{\textsc{#2} parameterized by~{#1}}%
                }
                \begin{tabular*}{\textwidth}{@{\hspace{.1em}} >{\itshape} p{1.4cm} p{0.85\textwidth} @{}}%
            }{
                \end{tabular*}%
                \end{GrayBox}%
                \ignorespacesafterend
            }

\newcommand{\defproblema}[3]{
  \begin{defproblemx}{#1}
    & \\[-8pt]
    Input:  & #2 \\[2pt]
    Task: & #3
  \end{defproblemx}
}%

\newcommand\problemdef[3]{
	\defproblema{#1}{#2}{#3}
}

\title{Graph Square Roots of Small Distance from Degree One Graphs\thanks{
A preliminary version of the paper has been accepted for LATIN 2020.
The paper received  support from the Research Council of Norway via the projects ``CLASSIS'' and ``MULTIVAL". }}

\author[1]{Petr A.\ Golovach}
\author[1]{Paloma~T.\ Lima}
\author[2]{Charis Papadopoulos}

\affil[1]{Department of Informatics, University of Bergen, Norway\\
 \texttt{\{petr.golovach,paloma.lima\}@uib.no}}

\affil[2]{Department of Mathematics, University of Ioannina, Greece\\
 \texttt{charis@uoi.gr}}

\date{}

\begin{document}


\maketitle

\begin{abstract}
Given a graph class $\mathcal{H}$, the task of the  $\mathcal{H}$-{\sc Square Root} problem is to decide, whether an input graph $G$ has a square root $H$ from $\mathcal{H}$. We are interested in  the parameterized complexity of the  problem for classes $\mathcal{H}$ that are composed by the graphs at vertex deletion distance at most $k$ from graphs of maximum degree at most one, that is, we are looking for a square root $H$ such that there is a modulator $S$ of size $k$ such that $H-S$ is the disjoint union of isolated vertices and disjoint edges.
 We show that different variants of the problems with constraints on the number of isolated vertices and edges in $H-S$ are \FPT{} when parameterized by $k$ by demonstrating algorithms with running time $2^{2^{\Oh(k)}}\cdot n^{\Oh(1)}$.
 We further show that the running time of our algorithms is asymptotically optimal and it is unlikely that the double-exponential dependence on $k$ could be avoided.
 In particular, we prove that the \sroot{} problem, that asks whether an input graph has a square root with vertex cover of size at most $k$, cannot be solved in time $2^{2^{o(k)}}\cdot n^{\Oh(1)}$ unless Exponential Time Hypothesis fails.
 Moreover, we point out 
that \sroot{} parameterized by $k$ does not admit a subexponential kernel unless $\P=\NP$.
\end{abstract}

\section{Introduction}

Squares of graphs and square roots constitute widely studied concepts in graph theory, both from a structural perspective as well as from an algorithmic point of view. A graph $G$ is the \emph{square} of a graph $H$ if $G$ can be obtained from $H$ by the addition of an edge between any two vertices of $H$ that are at distance two. In this case, the graph $H$ is called a \emph{square root} of $G$. It is interesting to notice that there are graphs that admit different square roots, graphs that have a unique square root and graphs that do not have a square root at all. In 1994, Motwani and Sudan~\cite{MotwaniS94} proved that the problem of determining if a given graph $G$ has a square root is \NP-complete. This problem is known as the {\sc Square Root} problem.

The intractability of {\sc Square Root} has been attacked in two different ways. The first one is by imposing some restrictions on the input graph $G$. In this vein, the {\sc Square Root} problem has been studied in the setting in which $G$ belongs to a specific class of graphs~\cite{CochefertCGKPS18,GKPS16b,GolovachKPS17,LeT10,MilanicS13,MOS14,NT14}.

Another way of coping with the hardness of the {\sc Square Root} problem is by imposing some additional structure on the square root $H$. That is, given the input graph $G$, the task is to determine whether $G$ has a square root $H$ that belongs to a specific graph class $\mathcal{H}$. This setting is known as the {\sc $\mathcal{H}$-Square Root} problem and it is the focus of this work.
%
%
The {\sc $\mathcal{H}$-Square Root} problem has been shown to be polynomial-time solvable for specific graph classes $\mathcal{H}$ \cite{LauC04,LeT10,LeT11,LOS15,LOS}.
To name a few among others, the problem is solved in polynomial time when $\mathcal{H}$ is the class of
trees~\cite{LinS95},
bipartite graphs~\cite{Lau06},
cactus graphs~\cite{GKPS16b},
and, more recently, when $\mathcal{H}$ is the class of
cactus block graphs~\cite{Ducoffe19},
outerplanar graphs~\cite{GHKLP19}, and
graphs of pathwidth at most~2~\cite{GHKLP19}.
It is interesting to notice that the fact that
 $\mathcal{H}$-{\sc Square Root} 
can be efficiently (say, polynomially)  solved
for some class~${\cal H}$ does not automatically imply
 that  $\mathcal{H'}$-{\sc Square Root} is
efficiently
solvable for every subclass
$\mathcal{H}'$ of~$\mathcal{H}$.
On the negative side, {\sc ${\cal H}$-Square Root} remains \NP-complete on graphs of girth at least~5~\cite{FarzadK12},
graphs of girth at least~4~\cite{FarzadLLT12},
split graphs~\cite{LauC04}, and
chordal graphs~\cite{LauC04}.
The fact that all known \NP-hardness constructions involve dense graphs \cite{FarzadK12,FarzadLLT12,LauC04,MotwaniS94} and dense square roots, raised the question of whether
{\sc ${\cal H}$-Square Root} is polynomial-time solvable for every sparse graph class $\mathcal{H}$.

We consider this question from the Parameterized Complexity viewpoint for structural parameterizations of $\mathcal{H}$ (we refer to the book of Cygan et al.~\cite{CyganFKLMPPS15} for an introduction to the field). More precisely, we are interested in graph classes $\mathcal{H}$ that are at {\it small distance} from a (sparse) graph class for which  {\sc ${\cal H}$-Square Root} can be solved in polynomial time.
Within this scope, the distance is usually measured either by the number of edge deletions, edge additions or vertex deletions.
This approach for the problem was first applied by Cochefert et al. in~\cite{Cochefert0GKP16}, who considered {\sc ${\cal H}$-Square Root},
where $\mathcal{H}$ is the class of graphs that have a feedback edge set of size at most $k$, that is, for graphs that can be made forests by at most $k$ edge deletions.
They proved that {\sc ${\cal H}$-Square Root} admits a compression to a special variant of the problem with $\Oh(k^2)$ vertices, implying
that the problem can be solved in $2^{\Oh(k^4)}+\Oh(n^4m)$ time, i.e.,  is fixed-parameter tractable (\FPT) when parameterized by $k$.
Herein, we study whether the same complexity behavior occurs if we measure the distance by the number of {\it vertex deletions} instead of edge deletions.

Towards such an approach, the most natural consideration for {\sc ${\cal H}$-Square Root} is to ask for a square root of feedback \emph{vertex} set of size at most $k$.
The approach used by Cochefert et al.~\cite{Cochefert0GKP16} fails if ${\cal H}$ is the class of graphs that can be made forests by at most $k$ vertex deletions
and the question of the parameterized complexity of our problem for this case is open.
In this context, we consider herein the {\sc ${\cal H}$-Square Root} problem when $\mathcal{H}$ is the class of graphs of bounded vertex deletion distance to a disjoint union of isolated vertices and edges.
Our main result is that the problem is \FPT{} when parameterized by the vertex deletion distance.
Surprisingly, however, we conclude a notable difference on the running time compared to the edge deletion case even on such a relaxed variation:
a double-exponential dependency on the vertex deletion distance is highly unavoidable.
Therefore, despite the fact that both problems are \FPT, the vertex deletion distance parameterization for the {\sc ${\cal H}$-Square Root} problem requires substantial effort.
More formally, we are interested in the following problem.

\problemdef
	{\dpqroot}
	{A graph $G$ and nonnegative integers $p,q,k$ such that $p+2q+k=|V(G)|$.}
	{Decide whether there is a square root $H$ of $G$ such that $H - S$ is a graph isomorphic to $pK_1 + qK_2$, for a set $S$ on $k$ vertices.}

\noindent	
Note that when $q=0$,
the problem asks whether $G$ has a square root with a vertex cover of size (at most) $k$ and we refer to the problem as \sroot. If $p=0$, we obtain \textsc{Distance-$k$-to-Matching Square Root}.
Observe also that, given an algorithm solving \dpqroot, then by testing all possible values of $p$ and $q$ such that $p+2q=|V(G)|-k$,
 we can solve the \textsc{Distance-$k$-to-Degree-One Square Root} problem, whose task is to decide whether there is a square root $H$ such that the maximum degree of $H - S$ is at most one for a set $S$ on $k$ vertices. Note that a set of vertices $X$ inducing a graph of maximum degree one is known as a \emph{dissociation} set and the maximum size of a dissociation set is called the \emph{dissociation} number (see, e.g., \cite{YANNAKAKIS81}).  Thus, the task of  \textsc{Distance-$k$-to-Degree-One Square Root} is to find a square root $H$ with the dissociation number at least $|V(G)|-k$.


We show that \dpqroot{} can be solved in $2^{2^{\Oh(k)}}\cdot n^{\Oh(1)}$  time, that is, the problem is \FPT{} when parameterized by $k$, the size of the deletion set.
We complement this result by showing that the running time of our algorithm is asymptotically optimal in the sense that \sroot, i.e., the special case of \dpqroot{}
when $q=0$,   cannot be solved in  $2^{2^{o(k)}}\cdot n^{\Oh(1)}$ time unless \emph{Exponential Time Hypothesis} (\emph{ETH}) of Impagliazzo, Paturi and Zane~\cite{ImpagliazzoP01,ImpagliazzoPZ01} fails (see also~\cite{CyganFKLMPPS15} for an introduction to the algorithmic lower bounds based on ETH).
We also prove that \sroot{} does not admit a kernel of subexponential in $k$ size unless $\P=\NP$.

Motivated by the above results, we further investigate the complexity of the {\sc ${\cal H}$-Square Root} problem when $\mathcal{H}$ is the class of graphs of bounded deletion distance to a specific graph class. We show that the problem of testing whether a given graph has a square root of bounded deletion distance to a clique is also \FPT{} parameterized by the size of the deletion set.

\section{Preliminaries}

\paragraph{Graphs.}
All graphs considered here are finite undirected graphs without loops and multiple edges. We refer to the textbook by Bondy and Murty~\cite{Bondy} for any undefined graph terminology.
We denote the vertex set of $G$ by $V(G)$ and the edge set by $E(G)$.
We use $n$ to denote the number of vertices of a graph and use $m$ for the number of edges (if this does not create confusion). Given $x\in V(G)$, we denote by $N_G(x)$ the neighborhood of $x$. The closed neighborhood of $x$, denoted by $N_G[x]$, is defined as $N_H(x)\cup\{x\}$.  For a set $X\subset V(G)$, $N_G(X)$ denotes the set of vertices in $V(G)\setminus X$ that have at least one neighbor in $X$. Analogously, $N_G[X]=N_G(X)\cup X$. The \emph{distance} between a pair of vertices $u,v\in V(G)$ is the number of edges of a shortest path between them in~$G$. We denote by $N^2_G(u)$ the set of vertices of $G$ that are at distance \emph{exactly} two from~$u$, and $N_G^2[u]$ is the set of vertices at distance at most two from $u$.
Given $S\subseteq V(G)$, we denote by $G-S$ the graph obtained from $G$ by the removal of the vertices of $S$. If $S=\{u\}$, we also write $G-u$. The \emph{subgraph induced by $S$} is denoted by $G[S]$, and has $S$ as its vertex set and $\{uv~|~u,v\in S\mbox{ and }uv\in E(G)\}$ as its edge set. A \emph{clique} is a set $K\subseteq V(G)$ such that $G[K]$ is a complete graph. An \emph{independent set} is a set $I\subseteq V(G)$ such that $G[I]$ has no edges. A \emph{vertex cover} of $G$ is a set $S\subseteq V(G)$ such that $V(G)\setminus S$ is an independent set. A graph is \emph{bipartite} if its vertex set can be partitioned into two independent sets, say~$A$ and~$B$, and is \emph{complete bipartite} if it is bipartite and every vertex of $A$ is adjacent to every vertex of $B$. A \emph{biclique} in a graph $G$ is a set $B\subseteq V(G)$ such that $G[B]$ is a complete bipartite graph. A \emph{matching} in $G$ is a set of edges having no common endpoint.
We denote by $K_r$ the complete graph on $r$ vertices. Given two graphs $G$ and $G'$, we denote by $G+G'$ the disjoint union of them. For a positive integer $p$, $pG$ denotes the disjoint union of $p$ copies of $G$.

The \emph{square} of a graph $H$ is the graph $G=H^2$ such that $V(G)=V(H)$ and every two distinct vertices $u$ and $v$ are adjacent in $G$ if and only if they are at distance at most two in $H$. If $G=H^2$, then $H$ is a \emph{square root} of $G$.

Two vertices $u,v$ are said to be \emph{true twins} if $N_G[u]=N_G[v]$. A \emph{true twin class} of $G$ is a maximal set of vertices that are pairwise true twins. Note that the set of true twin classes of $G$ constitutes a partition of $V(G)$.
Let $\mathcal{T}=\{T_1,\ldots,T_r\}$. We define the \emph{prime-twin} graph $\mathcal{G}$ of $G$ as the graph with the vertex set $\mathcal{T}$ such that two distinct vertices $T_i$ and $T_j$ of $\mathcal{G}$ are adjacent if and only if $uv\in E(G)$ for $u\in T_i$ and $v\in T_j$.

\paragraph{Parameterized Complexity.}
We refer to the recent book of  \cite{CyganFKLMPPS15} for an introduction to Parameterized Complexity. Here we only state some basic definitions that are crucial for understanding.
In a \emph{parameterized problem}, each instance is supplied with an integer \emph{parameter} $k$, that is, each instance can be written as a pair $(I,k)$.
A parameterized problem is said to be \emph{fixed-parameter tractable} ({\FPT}) if it can be solved in time $f(k)\cdot |I|^{\Oh(1)}$ for some computable function~$f$.
A \emph{kernelization} for a parameterized problem is a polynomial time algorithm that maps each instance $(I,k)$ of a parameterized problem  to an instance $(I',k')$ of the same problem such that
(i) $(I,k)$ is a \yes-instance if and only if $(I',k')$ is a \yes-instance, and
(ii) $|I'|+k'$ is bounded by~$f(k)$ for some computable function~$f$.
The output $(I',k')$ is called a \emph{kernel}. The function~$f$ is said to be the \emph{size} of the kernel.

\paragraph{Integer Programming.}
We will use integer linear programming as a subroutine in the proof of our main result. In particular, we translate part of our problem as an instance of the following problem.

\problemdef
	{$p$-Variable Integer Linear Programming Feasibility}
	{An $m\times p$ matrix $A$ over $\mathbb{Z}$ and a vector $b\in \mathbb{Z}^{m}$.}
	{Decide whether there is a vector $x\in \mathbb{Z}^{p}$ such that $Ax\leq b$.}

Lenstra~\cite{Lenstra} and Kannan~\cite{Kannan} showed that the above problem is \FPT{} parameterized by $p$, while Frank and Tardos~\cite{FrankTardos} showed that this algorithm can be made to run also in polynomial space. We will make use of these results, that we formally state next.	

\begin{theorem}[\cite{FrankTardos,Kannan,Lenstra}]\label{thm:ILP}
{\sc $p$-Variable Integer Linear Programming Feasibility} can be solved using $\Oh(p^{2.5p+o(p)}\cdot L)$ arithmetic operations and space polynomial in $L$, where $L$ is the number of bits in the input.
\end{theorem}

\section{\dpqroot}\label{sec:main}

In this section we give an \FPT{} algorithm for the \dpqroot{} problem, parameterized by $k$. In the remainder of this section, we use $(G,p,q,k)$ to denote an instance of the problem.
Suppose that $(G,p,q,k)$ is a \yes-instance and $H$ is a square root of $G$ such that there is $S\subseteq V(G)$ of size $k$ and $H-S$ is isomorphic to $pK_1+qK_2$.
We say that $S$ is a \emph{modulator},
the $p$ vertices of $H - S$ that belong to $pK_1$ are called \emph{$S$-isolated} vertices
and the $q$ edges that belong to $qK_2$ are called \emph{$S$-matching} edges.  Slightly abusing notation, we also use these notions when $H$ is not necessarily a square root of $G$ but any graph such that $H-S$ has maximum degree one.

\subsection{Structural lemmas}
\label{sec:fptdonesroot}

We start by defining the following two equivalence relations on the set of ordered pairs of vertices of $G$.
Two pairs of adjacent vertices $(x,y)$ and $(z,w)$ are called \emph{matched twins}, denoted by $(x,y)\sim_{\text{mt}}(z,w)$, if the following conditions hold:
\begin{itemize}
\item $N_G[x] \setminus \{y\} = N_G[z] \setminus \{w\}$, and
\item $N_G[y] \setminus \{x\} = N_G[w] \setminus \{z\}$.
\end{itemize}
A pair of vertices $(x,y)$ is called \emph{comparable} if
$N_G[x] \subseteq N_G[y]$.
Two comparable pairs of vertices $(x,y)$ and $(z,w)$ are \emph{nested twins}, denoted by $(x,y) \sim_{\text{nt}} (z,w)$, if the following conditions hold:
\begin{itemize}
\item $N_G(x) \setminus \{y\} = N_G(z) \setminus \{w\}$, and
\item $N_G[y] \setminus \{x\} = N_G[w] \setminus \{z\}$.
\end{itemize}

We use the following properties of matched and nested twins.

\begin{lemma}\label{lem:alltwins}
Let $(x,y)$ and $(z,w)$ be two distinct pairs of adjacent vertices (resp.\ comparable pairs) of $G$ that are matched twins (resp.\ nested twins).
Then, the following holds:
\begin{enumerate}
\item $\{x,y\} \cap \{z,w\} = \emptyset$,

\item  $xw, zy \notin E(G)$,

\item  $yw \in E(G)$,

\item  if $(x,y)\sim_{\text{mt}}(z,w)$ then $xz \in E(G)$,

\item  if $(x,y)\sim_{\text{nt}}(z,w)$ then $xz \notin E(G)$,

\item $G-\{x,y\}$ and $G-\{z,w\}$ are isomorphic.
\end{enumerate}
\end{lemma}
\begin{proof}
For (i), we show that the end-vertices of both pairs are distinct.
It is not difficult to see that $(x,y) \nematch (y,x)$ and $(x,y) \nenest (y,x)$, since $x\in N_G[x] \setminus \{y\}$ and $x \notin N_G[y]\setminus \{x\}$. Assume, for the sake of contradiction, that the two pairs share one end-vertex.

\begin{itemize}
\item First, we show (i) for $\sim_{\text{mt}}$. Let  $(x,y) \sim_{\text{mt}} (z,w)$.
Suppose that $y=w$.  Then $x\notin N_G[y]\setminus \{x\}$ but $x\in N_G[w]\setminus \{z\}$, that is, $N_G[y]\setminus \{x\}\neq N_G[w]\setminus \{z\}$ contradicting  $(x,y) \sim_{\text{mt}} (z,w)$.
Assume that $y=z$. Then $z\in N_G[y]\setminus \{w\}$ but $z=y\notin N_G[x]\setminus\{y\}$; a contradiction.
The cases $x=z$ and $x=w$ are completely symmetric to the cases considered above.

\item Now we prove (i) for $ \sim_{\text{nt}}$. Let  $(x,y) \sim_{\text{nt}} (z,w)$.
Suppose that $y=w$. Then $x\notin N_G[y]\setminus \{x\}$ but $x\in N_G[w]\setminus \{z\}$, that is, $N_G[y]\setminus \{x\}\neq N_G[w]\setminus \{z\}$; a contradiction to $(x,y) \sim_{\text{nt}} (z,w)$.
Let $x=z$. Then $y\notin N_G(x)\setminus \{y\}$ but $y\in N_G[z]\setminus \{w\}$, and we get that  $N_G(x)\setminus \{y\}\neq N_G(z)\setminus \{w\}$, leading again to a contradiction.
Assume that  $y=z$.  Then $y=z\notin N_G[w]\setminus\{z\}$ but $y\in N_G[y]\setminus \{x\}$ and we again obtain a contradiction. The case $x=w$ is symmetric.
\end{itemize}
This completes the proof of (i). To show the remaining claims,
observe that $N_G[y] \setminus \{x\} = N_{G}[w]\setminus\{z\}$ holds in both relations.

For (ii), note that if $xw\in E(G)$, then $x\in N_G[w]\setminus\{z\}$ but $x\notin N_G[y]\setminus\{x\}$, a contradiction. So $xw\notin E(G)$. The same follows by a symmetric argument for the edge $yz$.

For (iii), note that if $yw\notin E(G)$, then $w\notin N_G[y]\setminus\{x\}$, but $w\in N_G[w]\setminus\{z\}$, a contradiction.

To show (iv), observe that if $xz\notin E(G)$, then $x\in N_G[x]\setminus\{y\}$, while $x\notin N_G[z]\setminus\{w\}$, a contradiction.

For (v), if $xz\in E(G)$, then $z\in N_G(x)\setminus\{y\}$, but $z\notin N_G(z)\setminus\{w\}$, a contradiction.

To see (vi), notice that $\{x,y\} \cap \{z,w\} = \emptyset$ by (i).  Consider $\alpha\colon V(G)\rightarrow V(G)$ such that $\alpha(x)=z$, $\alpha(y)=w$, $\alpha(z)=x$, $\alpha(w)=y$ and $\alpha(v)=v$ for $v\in V(G)\setminus\{x,y,z,w\}$. It is straightforward to see that $\alpha$ is an automorphism of $G$ by the definition of $\enest$ and $\ematch$ and the properties (i) and (ii). Hence, $G-\{x,y\}$ and $G-\{z,w\}$ are isomorphic.
\end{proof}

In particular, the properties above allow us to classify pairs of vertices with respect to $\sim_{\text{mt}}$ and $\sim_{\text{nt}}$.

\begin{observation}\label{lem:equiv}
The relations $\sim_{\text{mt}}$ and $\sim_{\text{nt}}$ are equivalence relations on pairs of vertices and comparable pairs of vertices, respectively.
\end{observation}

\begin{proof}
It is clear that $\ematch$ (resp.\ $\enest$) are reflexive and symmetric on pairs of vertices (resp.\ comparable vertices).
Let $(x_1,y_1)$, $(x_2,y_2)$ and $(x_3,y_3)$ be pairs of vertices.
If $N_G[x_1]\setminus\{y_1\}=N_G[x_2]\setminus\{y_2\}$ and $N_G[x_2]\setminus\{y_2\}=N_G[x_3]\setminus\{y_3\}$, then
$N_G[x_1]\setminus\{y_1\}=N_G[x_3]\setminus\{y_3\}$. Also if
$N_G(x_1)\setminus\{y_1\}=N_G(x_2)\setminus\{y_2\}$ and $N_G(x_2)\setminus\{y_2\}=N_G(x_3)\setminus\{y_3\}$, then
$N_G(x_1)\setminus\{y_1\}=N_G(x_3)\setminus\{y_3\}$. This immediately implies that $\sim_{\text{mt}}$ and $\sim_{\text{nt}}$ are transitive, as well.
%
%
%
%
%
\end{proof}

\begin{figure}[ht]
\centering
\scalebox{0.8}{
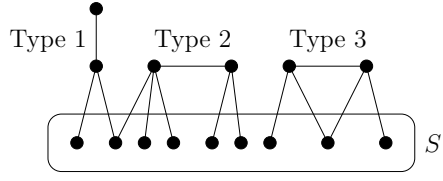}
\caption{Types of edges of $H-S$.}
\label{fig:typesofesdges}
\end{figure}

Let $H$ be a square root of a connected graph $G$ with at least three vertices, such that $H$ is at distance $k$ from $pK_1+qK_2$, and let $S$ be a modulator. Note that $S\neq\emptyset$, because $G$ is connected and $|V(G)|\geq 3$.
Then an $S$-matching edge $ab$ of $H$ satisfies exactly one of the following conditions:
\begin{itemize}
\item[1.] $N_H(a)\cap S=\emptyset$ and $N_H(b)\cap S\neq\emptyset$, 
\item[2.] $N_H(a)\cap S, N_H(b)\cap S\neq\emptyset$ and $N_H(a)\cap N_H(b)\cap S =\emptyset$,
\item[3.] $N_H(a)\cap S, N_H(b)\cap S\neq\emptyset$ and $N_H(a)\cap N_H(b)\cap S \neq\emptyset$.
\end{itemize}
We refer to them as type 1, 2 and~3 edges, respectively (see Figure~\ref{fig:typesofesdges}). We use the same notation for every graph $F$ that has a set of vertices $S$ such that $F-S$ has maximum degree at most one.

In the following three lemmas, we show the properties of the $S$-matching edges of types~1, 2 and 3 respectively that are crucial for our algorithm. We point out that even though some of the properties presented may   be  redundant, we state them in the lemmas for clarity of the explanations.

\begin{lemma}\label{lem:bigQset}
Let $H$ be a square root of a connected  graph $G$ with at least three vertices such that $H-S$ is isomorphic to $pK_1+qK_2$ for $S\subseteq V(G)$. If $a_1b_1$ and $a_2b_2$ are two type 1 distinct edges such that $N_H(b_1)\cap S=N_H(b_2)\cap S\neq\emptyset$, then the following holds:
\begin{enumerate}
\item $(a_1,b_1)$ and $(a_2,b_2)$ are comparable pairs,
\item $(a_1,b_1) \enest (a_2,b_2)$,
\item $(a_1,b_1) \nematch (a_2,b_2)$.
\end{enumerate}
\end{lemma}

\begin{proof}
Let $A=N_H(b_1)\cap S=N_H(b_2)\cap S$.
Since  $(a_1,b_1)$ is a type 1 edge, we have that $N_H(a_1)=\{b_1\}$. Thus, $N_G[a_1]=A\cup\{a_1,b_1\}\subseteq N_H[b_1]\subseteq N_G[b_1]$.
 The same holds for $(a_2,b_2)$. Hence, the pairs are comparable  and (i) is proved.

For (ii), note that since $N_H(b_1)\cap S=N_H(b_2)\cap S=A$, then $N_G[b_1]\setminus\{a_1\}=N_H[A]=N_G[b_2]\setminus\{a_2\}$.
Moreover, since $N_H(a_1)=\{b_1\}$ and $N_H(a_2)=\{b_2\}$, we have that $N_G(a_1)\setminus\{b_1\}=A=N_G(a_2)\setminus\{b_2\}$. This shows that $(a_1,b_1) \enest (a_2,b_2)$.

Finally, for (iii), it suffices to notice that by Lemma~\ref{lem:alltwins}(v), we have that $a_1a_2\notin E(G)$ and it should be $a_1a_2\in E(G)$ if $(a_1,b_1) \ematch (a_2,b_2)$ by Lemma~\ref{lem:alltwins}(iv).
\end{proof}

\begin{lemma}\label{lem:bigSset}
Let $H$ be a square root of a  connected  graph $G$ with at least three vertices such that $H-S$ is isomorphic to $pK_1+qK_2$ for $S\subseteq V(G)$. If $a_1b_1$ and $a_2b_2$ are two distinct type 2 edges such that $N_H(a_1)\cap S=N_H(a_2)\cap S$ and $N_H(b_1)\cap S=N_H(b_2)\cap S$, then the following holds:
\begin{enumerate}
\item $(a_1,b_1) \ematch (a_2,b_2)$,
\item $(a_1,b_1) \nenest (a_2,b_2)$.
\end{enumerate}
\end{lemma}

\begin{proof}
Let $A=N_H(a_1)\cap S=N_H(a_2)\cap S$ and $B=N_H(b_1)\cap S=N_H(b_2)\cap S$. Since $a_1b_1$ and $a_2b_2$ are type 2 edges, we have $A\cap B=\emptyset$.


For (i), notice that $N_G[a_1]=N^2_H[a_1]=\{b_1\}\cup B\cup N_H[A]$. Therefore, we have that $N_G[a_1]\setminus\{b_1\}=N_H^2[a_1]\setminus\{b_1\}=B\cup N_H[A]$. By the same arguments,
 $N_G[a_2]\setminus\{b_2\}=B\cup N_H[A]$ and, therefore,  $N_G[a_1]\setminus\{b_1\}=N_G[a_2]\setminus\{b_2\}$. By symmetric arguments, we obtain that $N_G[b_1]\setminus\{a_1\}=N_G[b_2]\setminus\{a_2\}$, which completes the proof  that $(a_1,b_1) \ematch (a_2,b_2)$.

To prove (ii), notice that $a_1a_2\in E(G)$ by Lemma~\ref{lem:alltwins}(iv) and, therefore, $(a_1,b_1) \nenest (a_2,b_2)$ by Lemma~\ref{lem:alltwins}(v).
\end{proof}

\begin{lemma}\label{lem:truetwins}
Let $H$ be a square root of a  connected  graph $G$ with at least three vertices such that $H-S$ is isomorphic to $pK_1+qK_2$ for $S\subseteq V(G)$. If $a_1b_1$ and $a_2b_2$ are two distinct type 3 edges such that $N_H(a_1)\cap S=N_H(a_2)\cap S$ and $N_H(b_1)\cap S=N_H(b_2)\cap S$, then the following holds:
\begin{enumerate}
\item $(a_1,b_1) \nematch (a_2,b_2)$,
\item $(a_1,b_1) \nenest (a_2,b_2)$,
\item $a_1$ and $a_2$ (resp.\ $b_1$ and $b_2$) are true twins in $G$.
\end{enumerate}
\end{lemma}

\begin{proof}
Let $A=N_H(a_1)\cap S=N_H(a_2)\cap S$ and $B=N_H(b_1)\cap S=N_H(b_2)\cap S$. Since $a_1b_1$ and $a_2b_2$ are type 3 edges, $A\cap B\neq\emptyset$.

For (i) and (ii), it suffices to notice that since $A\cap B\neq\emptyset$, then $a_1b_2,b_1a_2\in E(G)$. By Lemma~\ref{lem:alltwins}(ii), we conclude that $(a_1,b_1) \nematch (a_2,b_2)$ and $(a_1,b_1) \nenest (a_2,b_2)$.

For (iii), observe that $N_G[a_1]=N_H^2[a_1]=N_H[A]\cup\{b_1\}\cup B$ by the definition. Since $A\cap B\neq\emptyset$, we have that $b_1\in N_H[A]$. Hence,
$N_G[a_1]=N_H[A]\cup B$. By the same arguments,  $N_G[a_2]=N_H[A]\cup B$. Then $N_G[a_1]=N_G[a_2]$, that is, $a_1$ and $a_2$ are true twins. Clearly, the same holds  for $b_1$ and $b_2$.
\end{proof}

We also need the following straightforward observation about $S$-isolated vertices.

\begin{observation}\label{obs:isolated}
Let $H$ be a square root of a  connected  graph $G$ with at least three vertices such that $H-S$ is isomorphic to $pK_1+qK_2$ for $S\subseteq V(G)$. Then every two distinct $S$-isolated vertices of $H$ with the same neighbors in $S$ are true twins in $G$.
\end{observation}

The next lemma is used to construct reduction rules that allow to bound the size of equivalence classes of pairs of vertices with respect to $\enest$ and $\ematch$.

\begin{lemma}\label{lem:twonestedpairs}
Let $H$ be a square root of a  connected  graph $G$ with at least three vertices such that $H-S$ is isomorphic to $pK_1+qK_2$ for a modulator $S\subseteq V(G)$ of size $k$.
Let $Q$ be  an equivalence class in the set of pairs of comparable pairs of vertices with respect to the relation $\enest$ (an equivalence class in the set of pairs of adjacent vertices with respect to the relation $\ematch$, respectively).   If $|Q|\geq 2k + 2^{2k} + 1$, then $Q$ contains two pairs $(a_1,b_1)$ and  $(a_2,b_2)$ such that $a_1b_1$ and $a_2b_2$ are $S$-matching edges of type~1 in $H$ satisfying $N_H(b_1)\cap S=N_H(b_2)\cap S\neq \emptyset$ ($S$-matching edges of type~2 in $H$ satisfying $N_H(a_1)\cap S=N_H(a_2)\cap S$ and $N_H(b_1)\cap S=N_H(b_2)\cap S$, respectively).
\end{lemma}

\begin{proof}
Let $Q$ be an equivalence class of size at least $2k + 2^{2k} + 1$ with respect to $\enest$ or $\ematch$.

By Lemma~\ref{lem:alltwins}~(i), each vertex of $G$ appears in at most one pair of $Q$. Since $|S|=k$, there are at most $k$ pairs of $Q$ with at least one element in $S$.
 Let $$Q'=\{(x,y)\in Q\mid x,y\notin S\mbox{ and $xy$ is not a $S$-matching edge in $H$}\}.$$ We now show that $|Q'|\leq k$. Consider $(x,y)\in Q'$. Since $xy\in E(G)\setminus E(H)$, there exists $w\in V(G)$ such that $wx,wy\in E(H)$. Since $H - S$ 
 is isomorphic to $pK_1+qK_2$, we have that $w\in S$. Let $(x',y')\in Q'\setminus\{(x,y)\}$. 
 By the same argument, there exists $w'\in S$ such that $w'x',w'y'\in E(H)$. Moreover, it cannot be the case that $w=w'$, 
 since this would imply that  $xy',yx'\in E(G)$, which by Lemma~\ref{lem:alltwins}~(ii) is a contradiction to the fact that $(x,y) \enest (x',y')$ or $(x,y) \ematch (x',y')$. That is, for each pair $(x,y)\in Q'$, there is a vertex in $S$ that is adjacent to both elements of the pair and no vertex of $S$ can be adjacent to the elements of more than one pair of $Q'$.
Since $|S|\leq k$, we conclude that $|Q'|\leq k$.

Since $|Q|\geq 2k + 2^{2k} + 1$, there are at least $2^{2k}+1$ $S$-matching edges in $Q$. Given that $|S|\leq k$, by the pigeonhole principle, we have that there are two pairs $(a_1,b_1),(a_2,b_2)\in Q$ such that $a_1b_1$ and $a_2b_2$ are $S$-matched edges in $H$ and $N_H(a_1)\cap S=N_H(a_2)\cap S$ and $N_H(b_1)\cap S=N_H(b_2)\cap S$. In particular, this implies that $a_1b_1$ and $a_2b_2$ are of the same type. It cannot be the case that these two edges are of type 3, since by Lemma~\ref{lem:truetwins}(i) and (ii), these two pairs would not be equivalent with respect to $\enest$ or $\ematch$. We now consider the following two cases, one for each of the mentioned equivalence relations.

Suppose that $Q$ is an equivalence class in the set of pairs of comparable pairs of vertices with respect to the relation $\enest$.
By  Lemma~\ref{lem:bigSset}(ii), they cannot be of type 2.  
Hence, $a_1b_1$ and $a_2b_2$ are of type 1. In particular, either $N_H(a_1)\cap S=N_H(a_2)\cap S\neq \emptyset$
or $N_H(b_1)\cap S=N_H(b_2)\cap S\neq \emptyset$. If $N_H(a_1)\cap S=N_H(a_2)\cap S\neq \emptyset$, then $a_1a_2\in E(G)$ contradicting Lemma~\ref{lem:alltwins} (v). Hence,
$a_1b_1$ and $a_2b_2$ are $S$-matching edges of type~1 in $H$ satisfying $N_H(b_1)\cap S=N_H(b_2)\cap S\neq \emptyset$.

Let now $Q$ be an equivalence class in the set of pairs of adjacent vertices with respect to the relation $\ematch$.
By  Lemma~\ref{lem:bigQset}(iii), they cannot be of type 1.  
Hence, $a_1b_1$ and $a_2b_2$ are of type 2.
This concludes the proof of the lemma.
\end{proof}

\subsection{The algorithm}

In this section we prove our main result. First, we consider connected graphs.
For this, observe that if a connected graph $G$ has a square root $H$ then $H$ is connected as well.

\begin{theorem}\label{thm:mainfpt}
\dpqroot{} can be solved in time $2^{2^{\Oh(k)}}\cdot n^{\Oh(1)}$ on connected graphs.
\end{theorem}

\begin{proof}
Let $(G,p,q,k)$ be an instance of \dpqroot{} with $G$ being a connected graph. Recall that we want to determine if $G$ has a square root $H$ such that $H-S$ is isomorphic to $pK_1+qK_2$, for a modulator $S\subset V(G)$ with $|S|=k$, where $p+2q+k=n$. If $G$ has at most two vertices, then the problem is trivial. Notice also that if $k=0$, then $(G,p,q,k)$ may be a \yes-instance only if $G$ has at most two vertices, because $G$ is connected. Hence, from now on we assume that $n\geq 3$ and $k\geq 1$.

We exhaustively apply the following rule to reduce the number of type 1 edges in a potential solution. For this, we consider the set $\mathcal{A}$ of comparable pairs of vertices of $G$ and find its partition into equivalence classes with respect to $\enest$.  Note that $\mathcal{A}$ contains at most $2m$ elements and can be constructed in time $\Oh(mn)$. Then
the partition of $\mathcal{A}$ into equivalence classes can be found in time $\Oh(m^2n)$ by checking the neighborhoods of the vertices of each pair.

\begin{Rule}\label{rule:type1}
If there is an equivalence class $Q\subseteq \mathcal{A}$ with respect to $\enest$ such that $|Q|\geq 2k+2^{2k}+2$, delete two vertices of $G$ that form a pair of $Q$ and set $q:=q-1$.
\end{Rule}

The following claim shows that Rule~\ref{rule:type1} is safe.

\begin{claim}\label{claim:rule1safe}
If $G'$ is the graph obtained from $G$ by the application of Rule~\ref{rule:type1}, then $(G,p,q,k)$ and $(G',p,q-1,k)$ are equivalent instances of \dpqroot{} and $G'$ is connected.
\end{claim}

\begin{claimproof}
Let $G'=G-\{x,y\}$ for a pair $(x,y)\in Q$.

First assume $(G,p,q,k)$ is a \yes-instance to \dpqroot{} and let $H$ be a square root of $G$ that is a solution to this problem with a modulator $S$. By Lemma~\ref{lem:twonestedpairs}, $H$ has two $S$-matching edges $x'y'$ and $x''y''$ of type~1 such that $(x',y'),(x'',y'')\in Q$ and $N_H(y')\cap S=N_H(y'')\cap S\neq\emptyset$. Note that any edge of $H$ that has an endpoint in $y'$, also has an endpoint in $y''$ (except for $x'y'$). Hence, $H'=H-\{x',y'\}$ is a square root of $G''=G-\{x',y'\}$ with one less $S$-matching edge. Moreover, $H'$ is connected, because $H$ is connected and $N_H(y')\setminus\{x'\}=N_{H'}(y'')\setminus\{x''\}$. This implies that $G''$ is connected as well.
We conclude that $(G'',p,q-1,k)$ is a \yes-instance with $G''$ be a connected graph.
Because $G'$ and $G''$ are isomorphic by Lemma~\ref{lem:alltwins}(vi), we have that $(G',p,q-1,k)$ is a \yes-instance as well and $G'$ is connected.

Now assume $(G',p,q-1,k)$ is a \yes-instance to \dpqroot{} and let $H'$ be a square root of $G'$ that is a solution to this problem with a modulator $S$. Recall that $Q$ consists of pairs of vertices whose end-vertices are pairwise distinct by Lemma~\ref{lem:alltwins}(i). Hence,  $Q'=Q\setminus\{(x,y)\}$ contains at least $2k+2^{2k}+1$ elements.
By the definition of $\enest$, every two pairs of  $Q'=Q\setminus\{(x,y)\}$ are equivalent with respect to the relation for $G'$.
Thus, by Lemma~\ref{lem:twonestedpairs}, there are $(x',y'),(x'',y'')\in Q'$ such that $x'y'$ and $x''y''$ are $S$-matching edges of type~1 in $H'$ and $N_H(y')\cap S=N_H(y'')\cap S\neq\emptyset$. We construct a square root $H$ for $G$ by adding the edge $xy$ to $H'$ as an $S$-matching edge of type~1 with $N_H(y)\cap S=N_H(y')\cap S$.
To see that $H$ is indeed a square root for $G$, note that since $H'$ was a square root for $G'$, we have $H'^2=G'$. Now we argue about the edges of $G$ that are incident to $x$ and $y$. Since $(x,y),(x',y')\in Q$, we have that $N_G(x)\setminus \{y\}=N_G(x')\setminus \{y'\}$ and $N_G[y]\setminus \{x\}=N_G[y']\setminus\{x'\}$. This means that if $w\neq x$ is a neighbor of $y$ in $G$, then $w$ is also a neighbor of $y'$. Since $H'$ is a square root of $G'$, we have that either $y'w\in E(H')$ or $y'$ and $w$ are at distance two in $H'$. Since $N_H(y)\cap S=N_H(y')\cap S$, the same holds for $y$: it is either adjacent to $w$ or it is at distance two from $w$ in $H$. A symmetric argument holds for any edge incident to $x$ in $G$. Hence, we conclude that $H$ is indeed a square root of~$G$.
\end{claimproof}

We also want to reduce the number of type 2 edges in a potential solution. Let $\mathcal{B}$ be the set of pairs of adjacent vertices. We construct the partition of $\mathcal{B}$ into equivalence classes with respect to $\ematch$.  We have that $|\mathcal{B}|=2m$ and, therefore,
the partition of $\mathcal{B}$ into equivalence classes can be found in time $\Oh(m^2n)$ by checking the neighborhoods of the vertices of each pair.
 We exhaustively apply the following rule.

\begin{Rule} \label{rule:type2}
If there is an equivalence class $Q\subseteq \mathcal{B}$ with respect to $\ematch$ such that $|Q|\geq 2k+2^{2k}+2$, delete two vertices of $G$ that form a pair of $Q$ and set $q:=q-1$.
\end{Rule}

The following claim shows that Rule~\ref{rule:type1} is safe.

\begin{claim}\label{claim:rule2safe}
If $G'$ is the graph obtained from $G$ by the application of Rule~\ref{rule:type2}, then $(G,p,q,k)$ and $(G',p,q-1,k)$ are equivalent instances of \dpqroot{} and $G'$ is connected.
\end{claim}

\begin{claimproof}
The proof of this claim follows the same lines as the proof of Claim~\ref{claim:rule1safe}. Let $G'=G-\{x,y\}$ for $(x,y)\in Q$.

Let $(G,p,q,k)$ be a \yes-instance to \dpqroot{} and let $H$ be a square root of $G$ that is a solution to this problem with a modulator $S$. By Lemma~\ref{lem:twonestedpairs}, $H$ has two $S$-matching edges $x'y'$ and $x''y''$ of type~2 such that $(x',y'),(x'',y'')\in Q$ and $N_H(x')\cap S=N_H(x'')\cap S$ and $N_H(y')\cap S=N_H(y'')\cap S$.
Note that any edge of $H$ that has an endpoint in $x'$ (resp.\ $y'$), also has an endpoint in $x''$ (resp.\ $y''$), except for $x'y'$ (resp.\ $x''y''$). Thus, $H'=H-\{x',y'\}$ is a square root of $G''=G-\{x',y'\}$ with one less $S$-matching edge. We also have that
$H'$ is connected, because $H$ is connected. This implies that $G''$ is also connected.
We conclude that $(G'',p,q-1,k)$ is a \yes-instance with $G''$ be a connected graph.
Because $G'$ and $G''$ are isomorphic by Lemma~\ref{lem:alltwins} (vi), we have that $(G',p,q-1,k)$ is a \yes-instance as well and $G'$ is connected.

Now assume $(G',p,q-1,k)$ is a \yes-instance to  \dpqroot{} and let $H'$ be a square root of $G'$ that is a solution to this problem with a modulator $S$. Recall that $Q$ consists of pairs of vertices whose end-vertices are pairwise distinct by Lemma~\ref{lem:alltwins} (i). Hence,  $Q'=Q\setminus\{(x,y)\}$ contains at least $2k+2^{2k}+1$ elements.
By the definition of $\ematch$, every to pairs of  $Q'=Q\setminus\{(x,y)\}$ are equivalent with respect to the relation for $G'$.
Thus, by Lemma~\ref{lem:twonestedpairs}, there are $(x',y'),(x'',y'')\in Q'$ such that $x'y'$ and $x''y''$ are $S$-matching edges of type~2 in $H'$ with $N_H(x')\cap S=N_H(x'')\cap S$ and $N_H(y')\cap S=N_H(y'')\cap S$. We construct a square root $H$ for $G$ by adding the edge $xy$ to $H'$ as a $S$-matching edge of type~2 with $N_H(x)\cap S=N_H(x')\cap S$ and $N_H(y)\cap S=N_H(y')\cap S$.
To see that $H$ is indeed a square root for $G$, note that since $H'$ was a square root for $G'$, we have $H'^2=G'$. Now we argue about the edges of $G$ that are incident to $x$ and $y$. Since $(x,y),(x',y')\in Q$, we have that $N_G[x]\setminus \{y\}=N_G[x']\setminus \{y'\}$ and $N_G[y]\setminus \{x\}=N_G[y']\setminus\{x'\}$. This means that if $w\neq x$ is a neighbor of $y$ in $G$, then $w$ is also a neighbor of $y'$. Since $H'$ is a square root of $G'$, we have that either $y'w\in E(H')$ or $y'$ and $w$ are at distance two in $H'$. Since $N_H(y)\cap S=N_H(y')\cap S$, the same holds for $y$: it is either adjacent to $w$ in $H$ or it is at distance two from $w$ in $H$. A symmetric argument holds for any edge incident to $x$ in $G$. Hence, we conclude that $H$ is indeed a square root of~$G$.
\end{claimproof}

After exhaustive application of Rules~\ref{rule:type1} and~\ref{rule:type2} we obtain the following bounds on the number of $S$-matching edges of types~1 and~2 in a potential solution.

\begin{claim}\label{claim:boundedtype12}
Let $(G',p,q',k)$ be the instance of  \dpqroot{} after exhaustive applications of Rules~\ref{rule:type1} and~\ref{rule:type2}. Then $G'$ is a connected graph and a potential solution $H$ to the instance  has at most $2^k(2k+2^{2k}+1)$ $S$-matching edges of type~1 and $2^{2k}(2k+2^{2k}+1)$ $S$-matching edges of type~2.
\end{claim}

\begin{claimproof}
Clearly, $G'$ is connected by Claims~\ref{claim:rule1safe} and~\ref{claim:rule2safe}.

By Lemma~\ref{lem:bigQset}(ii) and Lemma~\ref{lem:bigSset}(ii), if two $S$-matching edges $xy$ and $x'y'$ of a potential solution behave in the same way with respect to $S$, that is, if $N_H(x)\cap S=N_H(x')\cap S$ and $N_H(y)\cap S=N_H(y')\cap S$, then they belong to the same equivalence class (either with respect to $\enest$ or to $\ematch$). Hence, after exhaustive application of Rule~\ref{rule:type1}, for each set $A\subseteq S$, there are at most $2k+2^{2k}+1$ $S$-matching edges $xy$ such that $N_H(y)\cap S=A$. Thus, there are at most $2^k(2k+2^{2k}+1)$ $S$-matching edges of type~1 in $H$. Analogously, after exhaustive application of Rule~\ref{rule:type1}, for each $A,B\subseteq S$, there are at most $2k+2^{2k}+1$ $S$-matching edges $xy$ such that $N_H(x)\cap S=A$ and $N_H(y)\cap S=B$. Hence, there are at most $2^{2k}(2k+2^{2k}+1)$ $S$-matching edges of type~2.
\end{claimproof}

For simplicity, we call $(G,p,q,k)$ again the instance obtained after exhaustive applications of Rules~\ref{rule:type1} and~\ref{rule:type2}. Notice that $G$ can be constructed in polynomial time, since the equivalence classes according to $\ematch$ and $\enest$ can be computed in time $\Oh(m^2n)$. 

By Claim~\ref{claim:boundedtype12}, in a potential solution, the number of $S$-matching edges of types~1 and~2 is bounded by a function of $k$. We will make use of this fact to make further guesses about the structure of a potential solution. To do so, we first consider the classes of true twins of $G$ and show the following.

\begin{claim}\label{claim:truetwinclasses}
Let $\mathcal{T}=\{T_1,\ldots,T_r\}$ be the partition of $V(G)$ into classes of true twins. If $(G,p,q,k)$ is a \yes-instance to our problem, then $r\leq 2(2^k+2^{2k})(2k+2^{2k}+1)+k+2^k+2\cdot2^{2k}$.
\end{claim}
\begin{claimproof}
Assume $(G,p,q,k)$ is a \yes-instance to our problem and let $H$ be a square root of $G$ containing a modulator $S$ of size $k$ such that $G - S$ is isomorphic to $pK_1+qK_2$. Let $X$ be the set of vertices of $G$ that are endpoints of type 1 and type 2 $S$-matching edges in $H$. By Claim~\ref{claim:boundedtype12}, $|X|\leq 2(2^k+2^{2k})(2k+2^{2k}+1)$. Note that if two $S$-isolated vertices of $H$ have the same neighborhood in $S$, they are true twins in $G$ by Observation~\ref{obs:isolated}.
Moreover, by Lemma~\ref{lem:truetwins}(iii), if $xy$ and $x'y'$ are two type~3 $S$-matching edges in $H$ satisfying $N_H(x)\cap S=N_H(x')\cap S$ and $N_H(y)\cap S=N_H(y')\cap S$, then $x$ and $x'$ (resp.\ $y$ and $y'$) are true twins in $G$. As already explained, there are no other types of edges in $H - S$.
Thus, we have at most $2(2^k+2^{2k})(2k+2^{2k}+1)$ distinct classes of true twins among the vertices of $X$, at most $k$ classes among the vertices of $S$, at most $2^k$ classes among the $S$-isolated vertices and at most $2\cdot 2^{2k}$ classes among the vertices that are endpoints of type~3 $k$-matching edges. This shows that $r\leq 2(2^k+2^{2k})(2k+2^{2k}+1)+k+2^k+2\cdot2^{2k}$.
\end{claimproof}

Observe that the partition $\mathcal{T}=\{T_1,\ldots,T_r\}$ of $V(G)$ into classes of true twins can be constructed in linear time~\cite{TEDDER08}. Using Claim~\ref{claim:truetwinclasses}, we apply the following rule.

\begin{Rule} \label{rule:twins}
If $|\mathcal{T}|>2(2^k+2^{2k})(2k+2^{2k}+1)+k+2^k+2\cdot2^{2k}$, then return \no{} and stop.
\end{Rule}

From now, we assume that we do not stop by Rule~\ref{rule:twins}. This means that $|\mathcal{T}|=\Oh(2^{4k})$.

Suppose that $(G,p,q,k)$ is a \yes-instance to \dpqroot{} and let $H$ be a square root of $G$ that is a solution to this instance with a modulator $S$. We say that $F$ is the \emph{skeleton} of $H$ with respect to $S$ if $F$ is obtained from $H$ be the exhaustive application of the following rules:
\begin{enumerate}
\item if $H$ has two distinct type~3 $S$-matching edges $xy$ and $x'y'$ with $N_H(x)\cap S=N_H(x')\cap S$ and $N_H(y)\cap S=N_H(y')\cap S$, then delete $x$ and $y$,
\item if $H$ has two distinct $S$-isolated vertices $x$ and $y$ with $N_H(x)=N_H(y)$, then delete $x$.
\end{enumerate}
In other word, we replace the set of $S$-matching edges of type~3 with the same neighborhoods on the end-vertices in $S$ by a single representative and we replace the set of $S$-isolated vertices with the same neighborhoods by a single representative.

We say that a graph $F$ is a \emph{potential solution skeleton} with respect to a set $S\subseteq V(F)$ of size $k$ for $(G,p,q,k)$ if the following conditions hold:
\begin{itemize}
\item[(i)] $F-S$ has maximum degree one, that is, $F-S$ is isomorphic to $sK_1+tK_2$ for some nonnegative integers $s$ and $t$,
\item[(ii)] for every two distinct $S$-isolated vertices $x$ and $y$ of $F$, $N_F(x)\neq N_F(y)$,
\item[(iii)] for every two distinct $S$-matching edges $xy$ and $x'y'$ of type~3, either $N_F(x)\cap S\neq N_H(x')\cap S$ or $N_F(y)\cap S\neq N_H(y')\cap S$,
\item[(iv)] for every $A,B\subseteq S$ such that $A\cap B=\emptyset$ and at least one of $A$ and $B$ is nonempty,
$\{xy\in E(F-S)\mid N_F(x)\cap S=A\text{ and }N_F(y)\cap S=B\}$ has size at most $2k+2^{2k}+1$.
\end{itemize}
Note that (iv) means that the number of type~1 and type~2 $S$-matched edges with the same neighbors in $S$ is upper bounded by  $2k+2^{2k}+1$.
Since Rules~\ref{rule:type1} and \ref{rule:type2} cannot be applied to $(G,p,q,k)$,
we obtain the following claim by
 Lemmas~\ref{lem:bigQset}(ii) and \ref{lem:bigSset}(ii).

\begin{claim}\label{claim:potential}
Every skeleton of a solution to $(G,p,q,k)$ is a  potential solution skeleton for this instance with respect to the modulator $S$.
\end{claim}

We observe that each potential solution skeleton has bounded size.

\begin{claim}\label{claim:potential-size}
For every potential solution skeleton $F$ for $(G,p,q,k)$, $$|V(F)|\leq k+2^k+2\cdot 2^{2k}+2\cdot 2^{2k}(2k+2^{2k}+1).$$
\end{claim}

\begin{claimproof}
By the definition, $F$ has $k$ vertices in $S$, at most $2^k$ $S$-isolated vertices and at most $2\cdot 2^{2k}$ end-vertices of $S$-matching edges of type~3. For each $A,B\subseteq S$ such that $A\cap B=\emptyset$ and at least one of $A$ and $B$ is nonempty,
$F$ has at most $k+2^{2k}+1$ $S$-matching edges $xy$ of type~1 or type~2 with $N_F(x)\cap S=A$  and $N_F(y)\cap S=B$. Then  we have at most $2\cdot 2^{2k}(2k+2^{2k}+1)$ end-vertices of these edges.
\end{claimproof}

Moreover, we can construct the family $\mathcal{F}$ of all potential solution skeletons together with their modulators.

\begin{claim}\label{claim:potential-fam-size}
The family $\mathcal{F}$ of all pairs $(F,S)$, where $F$ is a potential solution skeleton and $S\subseteq V(F)$ is a modulator of size $k$,
has size at most $2^{\binom{k}{2}}+2^{2^k}+2^{2^{2k}}+(2k+2^{2k}+2)^{2^{2k}}$ and
can be constructed in time $2^{2^{\Oh(k)}}$.
\end{claim}

\begin{claimproof}
There are at most $2^{\binom{k}{2}}$ distinct subgraph with the set of vertices $S$ of size $k$. We have at most $2^{2^k}$ distinct sets of $S$-isolated vertices and there are at most  $2^{2^{2k}}$  distinct sets of $S$-matching edges of type~3.  For each $A,B\subseteq S$ such that $A\cap B=\emptyset$ and at least one of $A$ and $B$ is nonempty, there are $k+2^{2k}+2$ possible sets of $S$-matching edges $xy$ of type~1 or type~2 such that $N_F(x)\cap S=A$ and $N_F(y)\cap S=B$. Therefore, we have at most $(2k+2^{2k}+2)^{2^{2k}}$ distinct sets of type~1 or type~2.
Then $|\mathcal{F}|\leq 2^{\binom{k}{2}}+2^{2^k}+2^{2^{2k}}+(2k+2^{2k}+2)^{2^{2k}}$.
Finally, it is straightforward to see that $\mathcal{F}$ can be constructed in $2^{2^{\Oh(k)}}$ time.
\end{claimproof}

Using Claim~\ref{claim:potential-fam-size}, we construct $\mathcal{F}$, and for every $(F,S)\in\mathcal{F}$, we check whether there is a solution $H$ to $(G,p,q,k)$  with a modulator $S'$, whose skeleton is isomorphic to $F$ with an isomorphism that maps $S$ to $S'$. If we find such a solution, then $(G,p,q,k)$ is a \yes-instance. Otherwise, Claims~\ref{claim:potential} guarantees that  $(G,p,q,k)$ is a \no-instance.

\medskip
Assume that we are given $(F,S)\in\mathcal{F}$ for the instance $(G,p,q,k)$.

Recall that we have the partition $\mathcal{T}=\{T_1,\ldots,T_r\}$ of $V(G)$ into true twin classes of  size at most $2(2^k+2^{2k})(2k+2^{2k}+1)+k+2^k+2\cdot2^{2k}$ by Rule~\ref{rule:twins}.
Recall also  that the prime-twin graph $\mathcal{G}$ of $G$ is the graph with the vertex set $\mathcal{T}$ such that two distinct vertices $T_i$ and $T_j$ of $\mathcal{G}$ are adjacent if and only if $uv\in E(G)$ for $u\in T_i$ and $v\in T_j$. Clearly, given $G$ and $\mathcal{T}$, $\mathcal{G}$ can be constructed in linear time.
For an induced subgraph $R$ of $G$, we define $\tau_R\colon V(R)\rightarrow \mathcal{T}$ to be a mapping such that $\tau_R(v)=T_i$ if $v\in T_i$ for $T_i\in\mathcal{T}$.

Let $\varphi\colon V(F)\rightarrow\mathcal{T}$ be a surjective mapping. We say that $\varphi$ is \emph{$\mathcal{G}$-compatible} if every two distinct vertices $u$ and $v$ of $F$ are adjacent in $F^2$ if and only if $\varphi(u)$ and $\varphi(v)$ are adjacent in $\mathcal{G}$.

\begin{claim}\label{claim:comp-skeleton}
Let $F$ be the skeleton of a solution $H$ to $(G,p,q,k)$. Then $\tau_F\colon V(F)\rightarrow \mathcal{T}$ is a $\mathcal{G}$-compatible surjection.
\end{claim}

\begin{claimproof}
Recall that $H^2=G$ and $F$ is an induced subgraph of $H$. Then the definition of $F$ and Lemma~\ref{lem:truetwins} (iii) immediately imply that $\tau_F$ is  a $\mathcal{G}$-compatible surjection.
\end{claimproof}

Our next step is to reduce our problem to solving a system of linear integer inequalities.
Let $\varphi\colon V(F)\rightarrow\mathcal{T}$ be a $\mathcal{G}$-compatible surjective mapping.
Let $X_1$, $X_2$ and $X_3$ be the sets of end-vertices of the $S$-matching edges of type~1, type~2 and type~3 respectively in $F$. Let also $Y$ be the set of $S$-isolated vertices of $F$.
For every vertex $v\in V(F)$, we introduce an integer variable $x_v$. Informally, $x_v$ is the number of vertices of a potential solution $H$ that correspond to a vertex $v$.
\begin{equation}\label{eq:syst}
\begin{cases}
x_v=1&\mbox{for }v\in S\cup X_1\cup X_2,\\
x_v\geq 1&\mbox{for }v\in Y\cup X_3,\\
x_u-x_v=0&\mbox{for every type~3 edge } uv,\\
\sum_{v\in Y}x_v=p,\\
\sum_{v\in X_1\cup X_2\cup X_3}x_v=2q,\\
\sum_{v\in\varphi^{-1}(T_i)}x_v=|T_i|&\mbox{for }T_i\in\mathcal{T}.
\end{cases}
\end{equation}

The following claim is crucial for our algorithm.

\begin{claim}\label{claim:system}
The instance $(G,p,q,k)$ has a solution $H$ with a modulator $S'$ such that  there is an isomorphism $\psi\colon V(F)\rightarrow V(F')$ for the skeleton $F'$ of $H$ mapping $S$ to $S'$  if and only if there is  a $\mathcal{G}$-compatible surjective mapping $\varphi\colon V(F)\rightarrow\mathcal{T}$ such that
the system (\ref{eq:syst}) has a solution.
\end{claim}

\begin{claimproof}
Suppose that there is a solution  $H$ to $(G,p,q,k)$ with a modulator $S'$, whose skeleton $F'$ is isomorphic to $F$  with an isomorphism that maps $S$ to $S'$.  To simplify notation, we identify $F$ and $F'$ and identify $S$ and $S'$. We set $\varphi=\tau_{F}$. By Claim~\ref{claim:comp-skeleton}, $\varphi$ is a $\mathcal{G}$-compatible surjection.
For $v\in Y$, we define the value of
$$x_v=|\{u\in V(H)\mid u\text{ is an }S\text{-isolated and }N_H(u)=N_H(v)\}|.$$
For each $S$-matching edge $uv$ of type~3 of $F$,
\begin{align*}
x_u=x_v=|\{xy\in E(H)\mid & xy\text{ is an }S\text{-matching edge},\\
&N_H(x)\cap S=N_H(u)\cap S\text{ and }N_H(y)\cap S=N_H(v)\cap S\}|.
\end{align*}
This defines the value of the variables $x_v$ for $v\in X_3$.
Recall that for all $v\notin Y\cup X_3$, $x_v=1$ by the definition of (\ref{eq:syst}).
It is straightforward to verify that the constructed assignment of the variables gives a solution of (\ref{eq:syst}) for $\varphi$.

For the opposite direction, let $\varphi\colon V(F)\rightarrow\mathcal{T}$ be a $\mathcal{G}$-compatible surjective mapping
such that the system (\ref{eq:syst}) has a solution. Assume that the variables $x_v$ have values that satisfy (\ref{eq:syst}).
We construct the graph $\hat{F}$ from $F$ and the extension $\hat{\varphi}$ of $\varphi$ as follows.
\begin{itemize}
\item For every $S$-isolated vertex $v$ of $F$, replace $v$ by $x_v$ copies that are adjacent to the same vertices as $v$ and define $\hat{\varphi}(x)=\varphi(v)$ for the constructed vertices.
\item For every $S$-matching edge $uv$ of type~3, replace $u$ and $v$ by $x_u=x_v$ copies of pairs of adjacent vertices $x$ and $y$, make $x$ and $y$ adjacent to the same vertices of $S$ as $u$ and $v$ respectively, and define $\hat{\varphi}(x)=\varphi(u)$ and $\hat{\varphi}(y)=\varphi(v)$ respectively.
\item Set $\hat{\varphi}(v)=\varphi(v)$ for the remaining vertices.
\end{itemize}
Observe that by the construction and the assumption that the values of the variables $x_v$ satisfy (\ref{eq:syst}), $\hat{F}$ has $p$ $S$-isolated vertices, $q$ $S$-matching edges, and for every $T_i\in\mathcal{T}$, $|\{v\in V(\hat{F})\mid v\in\hat{\varphi}^{-1}(T_i)\}|=|T_i|$.  We define $\psi\colon V(G)\rightarrow V(G)$ by mapping $|T_i|$ vertices of $\{v\in V(\hat{F})\mid v\in\hat{\varphi}^{-1}(T_i)\}$ arbitrarily into distinct vertices of $T_i\subseteq V(G)$ for each $T_i\in\mathcal{T}$. Clearly, $\psi$ is a bijection.
Notice that by Lemma~\ref{lem:truetwins} (iii) and Observation~\ref{obs:isolated}, the sets of vertices of $\hat{F}$ constructed from $S$-isolated vertices and the end-vertices of $S$-matching edges are sets of true twins in $\hat{F}^2$. Also we have that, because $\varphi$ is  $\mathcal{G}$-compatible, two distinct vertices $u,v\in V(\hat{F})$ are adjacent in $\hat{F}^2$ if and only if either $\hat{\varphi}(u)=\hat{\varphi}(v)$ or $\hat{\varphi}(u)\neq\hat{\varphi}(v)$ and $\hat{\varphi}(u)\hat{\varphi}(v)\in E(\mathcal{G})$. This implies that $\psi$ is an isomorphism of $\hat{F}^2$ and $G$, which means that $G$ has a square root isomorphic to $\hat{F}$. Clearly, $\psi|_{V(F)}$ is an isomorphism of $F$ into the skeleton of $H$ mapping $S$ to $S'$.
\end{claimproof}

By Claim~\ref{claim:system}, we can state  our task as follows: verify whether there is a $\mathcal{G}$-compatible surjection $\varphi\colon V(F)\rightarrow\mathcal{T}$ such that (\ref{eq:syst})
has a solution.

For this, we consider all at most $|V(F)|^{|\mathcal{T}|}=2^{2^{\Oh(k)}}$ surjections  $\varphi\colon V(F)\rightarrow\mathcal{T}$. For each $\varphi$, we verify whether it is $\mathcal{G}$-compatible.
Clearly,  it can be done in time $\Oh(|V(F)|^3)$. If $\varphi$ is $\mathcal{G}$-compatible, we construct the system (\ref{eq:syst}) with $|V(F)|=2^{\Oh(k)}$ variables in time $\Oh(|V(F)|^2)$. Then we solve it by applying Theorem~\ref{thm:ILP} in $2^{2^{\Oh(k)}}\log n$ time. This completes the description of the algorithm and its correctness proof.

To evaluate the total running time, notice that the preprocessing step, that is, the exhaustive application of Rules~\ref{rule:type1} and \ref{rule:type2} is done in polynomial time. Then the construction of $\mathcal{T}$, $\mathcal{G}$ and the application of Rule~\ref{rule:twins} is polynomial as well. By Claim~\ref {claim:potential-fam-size}, $\mathcal{F}$ is constructed in time $2^{2^{\Oh(k)}}$. The final steps, that is, constructing $\varphi$ and systems (\ref{eq:syst}) and solving the systems, can be done in time $2^{2^{\Oh(k)}}\log n$. Therefore, the total running time is $2^{2^{\Oh(k)}}\cdot n^{\Oh(1)}$.
\end{proof}

For simplicity, in Theorem~\ref{thm:mainfpt}, we assumed that the input graph is connected but it is not difficult to extend the result for general case.

\begin{corollary}\label{cor:mainfpt}
\dpqroot{} can be solved in time $2^{2^{\Oh(k)}}\cdot n^{\Oh(1)}$.
\end{corollary}

\begin{proof}
Let $(G,p,q,k)$ be an instance of \dpqroot and let $C_1,\ldots,C_\ell$ be the components of $G$. If $\ell=1$, we apply Theorem~\ref{thm:mainfpt}. Assume that this is not the case and $\ell\geq 2$. For each $i\in\{1,\ldots,\ell\}$, we use Theorem~\ref{thm:mainfpt} to solve the instances $(C_i,p',q',k')$ such that $k'\leq k$, $p'\leq p$, $q'\leq q$ and
$k'+p'+2q'=|V(C_i)|$. Then we combine these solutions to solve the input instance using a dynamic programming algorithm.

For $h\in\{1,\ldots,\ell\}$, let $G_h$ be the subgraph of $G$ with the components $C_1,\ldots,C_h$. Clearly, $G_h=G$. For each $h\in\{1,\ldots,\ell\}$, every triple of nonnegative integers $k',p',q'$ such that  $k'\leq k$, $p'\leq p$, $q'\leq q$ and
$k'+p'+2q'=|V(G_h)|$, we solve the instance $(G_h,p',q',k')$.
 For $h=1$, this is already done as $G_1=C_1$. Let $h\geq 2$. Then it is straightforward to observe that
$(G_h,p',q',k')$ is a \yes-instance if and only if there are nonnegative integers $k_1,p_1,q_1$ and $k_2,p_2,q_2$ such that
\begin{itemize}
\item  $k_1+k_2= k'$, $p_1+p_2= p'$, $q_1+q_2= q'$, and
\item $k_1+p_1+2q_1=|V(G_{h-1})|$ and $k_2+p_2+2q_2=|V(C_{h})|$,
\end{itemize}
for which both $(G_{h-1},p_1,q_1,k_1)$ and $(C_h,p_2,q_2,k_2)$ are \yes-instances.
This allows to solve $(G_h,p',q',k')$ in time $\Oh(n^2)$ if we are given the solutions for $(G_{h-1},p_1,q_1,k_1)$ and $(C_h,p_2,q_2,k_2)$.
We obtain that, given the tables of solutions for the components of $G$, we can solve the problem for $G$ in time $\Oh(n^5)$. We conclude that the total running time is $2^{2^{\Oh(k)}}\cdot n^{\Oh(1)}$.
\end{proof}

Corollary~\ref{cor:mainfpt} gives the following statement for the related problems.

\begin{corollary}\label{cor:VC}
\sroot{}, \textsc{Distance-$k$-to-Matching Square Root} and \textsc{Distance-$k$-to-Degree-One Square Root}
 can be solved in time $2^{2^{\Oh(k)}}\cdot n^{\Oh(1)}$.
\end{corollary}

\section{A lower bound for \dpqroot}
In this section, we show that the running time of our algorithm for \dpqroot{} given in Section~\ref{sec:main} (see Theorem~\ref{thm:mainfpt}) cannot be significantly improved.
In fact, we show that the \dpqroot{} problem admits a double-exponential lower bound, even for the special case $q=0$, that is, in the case of \sroot.


To provide a lower bound for the \sroot{} problem, we will give a parameterized reduction from the \bccover{} problem. This problem takes as input a bipartite graph $G$ and a nonnegative integer $k$, and the task is to decide whether the edges of~$G$ can be covered by at most $k$ complete bipartite subgraphs. Chandran et al.~\cite{CIK16} showed the following two results about the \bccover{} problem that will be of interest to us.

\begin{theorem}[\normalfont \cite{CIK16}]\label{theo:biclique}
\bccover~cannot be solved in time $2^{2^{o(k)}}\cdot n^{\Oh(1)}$ unless ETH is false.
\end{theorem}

\begin{theorem}[\normalfont \cite{CIK16}]\label{theo:bicliquekernel}
\bccover~does not admit a kernel of size $2^{o(k)}$ unless $\P=\NP$.
\end{theorem}

\begin{lemma}\label{lem:reduction}
There exists a polynomial time algorithm that, given an instance $(B,k)$ for \bccover, produces an equivalent instance $(G,k+4)$ for \sroot, with $|V(G)|=|V(B)|+k+6$.
\end{lemma}

\begin{proof}
Let $(B,k)$ be an instance of \bccover~where $(X,Y)$ is the bipartition of $V(B)$. Let $X=\{x_1,\ldots,x_p\}$ and $Y=\{y_1,\ldots,y_q\}$. We construct the instance $(G,k+4)$ for \sroot~such that $V(G)=X\cup Y\cup\{z_1,\ldots,z_k\}\cup\{u,v,w,u',v',w'\}$. Denote by $Z$ the set $\{z_1,\ldots,z_k\}$. The edge set of $G$ is defined in the following way: $G[X\cup Z\cup \{u\}]$, $G[X\cup \{v\}]$, $\{u,v,w\}$, $G[Y\cup Z\cup \{u'\}]$, $G[Y\cup \{v'\}]$ and $\{u',v',w'\}$ are cliques and $x_iy_j\in E(G)$ if and only if $x_iy_j\in E(B)$. The construction of $G$ is shown in Figure~\ref{fig:lower}.

\begin{figure}[ht]
\centering
\includegraphics[scale= 1.00]{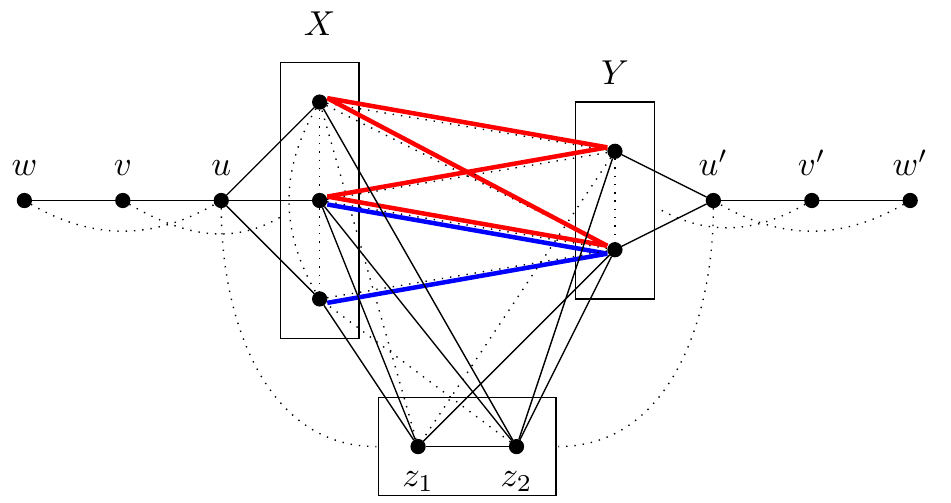}
\caption{Illustrating the graphs $G$ and $H$ considered in the proof of Lemma~\ref{lem:reduction}.
The sets $X$ and $Y$ form the bipartition of an instance of \bccover~and the two colored completed bipartite subgraphs correspond to a solution of the problem, where $k=2$.
The constructed graph $G$ of \sroot~ is depicted by the solid and dotted black edges, whereas
the graph spanned by the solid black edges corresponds to the square root $H$ of $G$.}
\label{fig:lower}
\end{figure}

For the forward direction, suppose $(B,k)$ is a \yes-instance for \bccover. We will show that $(G,k+4)$ is a \yes-instance for \sroot. Note that if $B$ has a biclique cover of size strictly less than $k$, we can add arbitrary bicliques to this cover and obtain a biclique cover for $B$ of size exactly $k$. Let $\mathcal{C}=\{C_1,\ldots,C_k\}$ be such a biclique cover. We construct the following square root candidate $H$ for $G$ with $V(H)=V(G)$. Add the edges $uv$, $vw$, $u'v'$ and $v'w'$ to $H$, and also all the edges between $u$ and $X$, all the edges between $u'$ and $Y$ and all the edges in $G[Z]$. Finally, for each $1\leq i\leq k$, add to $H$ all the edges between $z_i$ and the  vertices of $C_i$.

\begin{claim}\label{claim:srH}
The constructed graph $H$ is indeed a square root of $G$.
\end{claim}
\begin{claimproof}
Let $xy\in E(G)$ be such that $xy\notin E(H)$. If $xy=uw$, note that $uv,vw\in E(H)$. If $xy=vx_i$ for some $i$, note that $uv,ux_i\in E(G)$. If $xy=x_ix_j$, then $ux_i,ux_j\in E(H)$. If $xy=uz_i$, let $x_j$ be a vertex of $C_i$ and note that $ux_j,x_jz_i\in E(H)$. If $xy=x_jz_i$, let $\ell$ be such that $x_j\in C_\ell$ and observe that $x_jz_\ell,z_\ell z_i\in E(H)$. Symmetric arguments apply for the edges in $G[Y\cup Z\cup\{u',v',w'\}]$. Finally, if $xy=x_iy_j$, let $C_\ell$ be the biclique containing the edge $x_iy_j$ in $\mathcal{C}$ and note that $x_iz_\ell,y_jz_\ell \in E(H)$.
\end{claimproof}

We conclude that $(G,k+4)$ is a \yes-instance for \sroot~by Claim~\ref{claim:srH} together with the fact that $Z\cup\{u,v,u',v'\}$ is a vertex cover of $H$ of size $k+4$.

Before we show the reverse direction of the theorem, we state the next three claims, that concern the structure of any square root of the graph $G$.

\begin{claim}\label{claim:uvvw}
The edges $uv$, $vw$, $u'v'$ and $v'w'$ belong to any square root of $G$.
\end{claim}

\begin{claimproof}
Suppose for a contradiction that the graph $G$ has a square root $H$ such that $vw\notin E(H)$. In this case, it holds that $vu,uw\in E(H)$, since $u$ is the only common neighbor of $v$ and $w$. However, since $wx_i\notin E(G)$ for every $1\leq i\leq p$ and $wz_j\notin E(G)$ for every $1\leq j\leq k$, then $ux_i,uz_j\notin E(H)$. Therefore, there must exist an induced $P_3$ in $H$ with endpoints, for instance, $u$ and $z_\ell$, for some $\ell$. However, since $ux_i,uz_j\notin E(H)$ for every $1\leq i\leq p$ and every $1\leq j\leq k$ and $N_G(u)=X\cup Z\cup\{v,w\}$, either $v$ or $w$ have to be the middle vertex of the $P_3$. This is a contradiction, since $vz_\ell,wz_\ell\notin E(G)$.

Now suppose for a contradiction that the graph $G$ has a square root $H$ such that $uv\notin E(H)$. If there exists $\ell$ such that $ux_\ell,vx_\ell\in E(H)$, then we have a contradiction, since this would imply that no edge incident to $w$ can be in $H$, given that $wx_\ell\notin E(G)$. We can then conclude that $vw,uw\in E(H)$. We can now reach a contradiction by the same argument as used in the previous paragraph.

The claim follows by a symmetric argument for the edges $u'v'$ and $v'w'$.
\end{claimproof}

\begin{claim}\label{claim:uxi}
The edges $\{ux_i,u'y_j~|~1\leq i\leq p,1\leq j\leq q\}$ belong to any square root of $G$.
\end{claim}

\begin{claimproof}
Suppose for a contradiction that $G$ has a square root $H$ such that $ux_i\notin E(H)$ for some $1\leq i\leq p$. By Claim~\ref{claim:uvvw}, $uv,vw\in E(H)$. This implies that $vx_i\notin E(H)$, since $wx_i\notin E(G)$. Since $ux_i\notin E(H)$ by assumption, there must exist $j$ such that $x_j$ is the middle vertex of a $P_3$ in $H$ with endpoints $v$ and $x_i$. However, this is a contradiction, since $wx_j\notin E(G)$. The claim follows by a symmetric argument for the edges of the form $u'y_j$.
\end{claimproof}

\begin{claim}\label{claim:xiyj}
The edges $\{x_iy_j~|~1\leq i\leq p,1\leq j\leq q\}$ do not belong to any square root of $G$.
\end{claim}

\begin{claimproof}
Suppose for a contradiction that $G$ has a square root $H$ such that $x_iy_j\in E(H)$ for some $1\leq i\leq p$ and $1\leq j\leq q$. By Claim~\ref{claim:uxi}, we have that $ux_i\in E(H)$, which is a contradiction since $uy_j\notin E(G)$.
\end{claimproof}

Now, for the reverse direction of the theorem, assume that $G$ has a square root $H$ that has a vertex cover of size at most $k+4$. By Claim~\ref{claim:xiyj}, for every edge of $G$ of the form $x_iy_j$, it holds that $x_iy_j\notin E(H)$. This implies that, for every such edge, there exists an induced $P_3$ in $H$ having $x_i$ and $y_j$ as its endpoints. Since $N_G(x_i)\cap N_G(y_j)=Z$, only vertices of $Z$ can be the middle vertices of these paths. For $1\leq \ell\leq k$, let $C_\ell=N_H(z_\ell)\cap (X\cup Y)$. We will now show that $\mathcal{C}=\{C_1,\ldots,C_k\}$ is a biclique cover of $B$. First, note that since for every edge $x_iy_j$, there exists $z_h\in Z$ such that $z_hx_i,z_hy_j\in E(H)$, we conclude that $x_iy_j\in C_h$, which implies that $\mathcal{C}$ is an edge cover of $B$. Furthermore, for a given $\ell$, since every vertex of $C_\ell$ is adjacent to $z_\ell$ in $H$, $G[C_\ell]$ is a clique and, therefore, $B[C_\ell]$ is a biclique. This implies that $\mathcal{C}$ is indeed a biclique cover of $B$ of size $k$, which concludes the proof of the theorem.
\end{proof}

From Theorem~\ref{theo:biclique}  and Lemma~\ref{lem:reduction} we obtain the following theorems.

\begin{theorem}
 \sroot~cannot be solved in time $2^{2^{o(k)}}\cdot n^{\Oh(1)}$ unless ETH is false.
\end{theorem}

Moreover, from Theorem~\ref{theo:bicliquekernel} and Lemma~\ref{lem:reduction} we can also conclude the following corollary.

\begin{theorem}
\sroot~does not admit a kernel of size $2^{o(k)}$ unless $\P=\NP$.
\end{theorem}

\begin{proof}
Assume that \sroot~has a kernel of size $2^{o(k)}$. Since \sroot{} is in \NP{} and \bccover{}  is \NP-complete, there is an algorithm $\mathcal{A}$ that in time $\Oh(n^c)$ reduces  \sroot{} to \bccover{}, where $c$ is a positive constant. Then combining the reduction from Lemma~\ref{lem:reduction}, the kernelization  algorithm for \sroot{} and $\mathcal{A}$, we obtain a kernel for \bccover{}  of size $(2^{o(k)})^c$ that is subexponential in $k$. By   Theorem~\ref{theo:bicliquekernel}, this is impossible unless $\P=\NP$.
Equivalently, we can observe that Chandran et al.~\cite{CIK16}, in fact, proved a stronger claim. Their proof shows that \bccover~does not admit a \emph{compression} (we refer to~\cite{CyganFKLMPPS15} for the definition of the notion) of subexponential in $k$ size to any problem in \NP.
\end{proof}

\section{\cliqueroot}

In this section, we consider the complexity of testing whether a graph admits a square root of bounded deletion distance to a clique. More formally, we consider the following problem:

\problemdef
	{\cliqueroot}
	{A graph $G$ and nonnegative integer $k$.}
	{Decide whether there is a square root $H$ of $G$ such that $H - S$ is a complete graph for a set $S$ on $k$ vertices.}

We give an algorithm running in \FPT-time parameterized by $k$, the size of the deletion set. That is, we prove the following theorem.

\begin{theorem}\label{thm:cliqueroot}
\cliqueroot~can be solved in time $2^{2^{\Oh(k)}} \cdot n^{\Oh(1)}$.
\end{theorem}
\begin{proof}

Let $(G,k)$ be an instance to $\cliqueroot$. We start by computing the number of classes of true twins in $G$. If $G$ has at least $2^k+k+1$ classes of true twins, then~$G$ is a \no-instance to the problem, as we show in the following claim.

\begin{claim}
Let $G$ be a graph and $H$ be a square root of $G$ such that $H-S$ is a complete graph, with $|S|= k$. Let $T_1,\ldots,T_t$ be a partition of $V(G)$ into classes of true twins. Then $t\leq 2^k+k$.
\end{claim}

\begin{claimproof}
Let $C=V(H)\setminus S$. Note that if $u,v\in C$ and $N_H(u)\cap S= N_H(v)\cap S$, then $u$ and $v$ are true twins in $G$. Thus, we have at most $2^k$ distinct classes of true twins among the vertices of~$C$, and at most $k$ among the vertices of $S$.
\end{claimproof}

Hence, from now on we assume that $G$ has at most $2^k+k$ classes of true twins. We exhaustively apply the following rule in order to decrease the size of each class of true twins in $G$.

\begin{Rule}\label{rule:ttclique}
If $|T_i|\geq 2^k+k+1$ for some $i$, delete a vertex from $T_i$.
\end{Rule}

The following claim shows that Rule~\ref{rule:ttclique} is safe.

\begin{claim}
If $G'$ is the graph obtained from $G$ by the application of Rule~\ref{rule:ttclique}, then $(G,k)$ and $(G',k)$ are equivalent instances of \cliqueroot.
\end{claim}

\begin{claimproof}
Let $G'=G-v$. First assume $(G,k)$ is a \yes-instance to \cliqueroot{} and let $H$ be a square root of $G$ that is a solution to this problem. Since $|T_i|\geq 2^k+k+1$ and $G$ has at most $2^k+k$ classes of true twins, by the pigeonhole principle there are two vertices $x,y\in T_i$ such that, in $H$, $x,y\notin S$ and $N_H[x]\cap S=N_H[y]\cap S$. That is, $x$ and $y$ are true twins in $H$ also. Thus, $H'=H-x$ is a square root for $G''=G-x$ such that $H'-S$ is a complete graph. Since $G'$ and $G''$ are isomorphic, we have that $(G',k)$ is a \yes-instance as well.

Now assume $(G',k)$ is a yes-instance to \cliqueroot{} and let $H'$ be a square root of $G'$ that is a solution to the problem. Note that $T_i\setminus\{v\}$ is a true twin class of $G'$ of size at least $2^k+k$. Thus, there exists $u\in T_i\setminus\{v\}$ such that, in $H'$, $u\notin S$. We can add $v$ to $H'$ as a true twin of $u$ and obtain a square root $H$ for $G$ such that $H-S$ is a complete graph.
\end{claimproof}

After exhaustive application of Rule~\ref{rule:ttclique}, we obtain an instance $(G',p',k)$ such that $G'$ contains at most $(2^{k}+k)^2$ vertices, since it has at most $2^k+k$ twin classes, each of size at most $2^k+k$. Moreover, $(G',k)$ and $(G,k)$ are equivalent instances of \cliqueroot. We can now check by brute force whether $(G',k)$ is \yes-instance to the problem. Since $G'$ has $2^{\Oh(k)}$ vertices, this can be done in time $2^{2^{\Oh(k)}} \cdot n^{\Oh(1)}$, which concludes the proof of the theorem.
\end{proof}


\section{Conclusion}

In this work, we showed that \dpqroot{} and its variants can be solved in $2^{2^{\Oh(k)}}\cdot n^{\Oh(1)}$ time. We also proved that the double-exponential dependence on $k$ is unavoidable up to Exponential Time Hypothesis, that is, the problem cannot be solved in $2^{2^{o(k)}}\cdot n^{\Oh(1)}$ time unless ETH fails. We also proved that the problem does not admit a  kernel
of subexponential in $k$ size
unless  $\P=\NP$. We believe that it would be interesting to further investigate the parameterized complexity of  $\mathcal{H}$-{\sc Square Root} for sparse graph classes $\mathcal{H}$ under structural parameterizations. The natural candidates are the \textsc{Distance-$k$-to-Linear-Forest Square Root} and \textsc{Feedback-Vertex Set-$k$ Square Root} problems, whose tasks are to decide whether the input graph has a square root that can be made a linear forest, that is, a union of paths, and a forest respectively by (at most) $k$ vertex deletions. Recall that the existence
of an \FPT{} algorithm for $\mathcal{H}$-{\sc Square Root} does not imply the same for subclasses of $\mathcal{H}$. However, it can be noted that the reduction from Lemma~\ref{lem:reduction} implies that 
our complexity lower bounds still hold and, therefore, we cannot expect that these problems would be easier.

Parameterized complexity of $\mathcal{H}$-{\sc Square Root} is widely open for other, not necessarily sparse, graph classes. We considered the \cliqueroot{} problem and proved that it is \FPT{} when parameterized by $k$. What can be said if we ask for a square root that is at deletion distance (at most) $k$ form a \emph{cluster graph}, that is, the disjoint union of cliques? We believe that our techniques allows to show that this problem is \FPT{} when parameterized by $k$ if the number of cliques is a fixed constant. Is the problem \FPT{} without this constraint?

\bibliographystyle{siam}
\bibliography{ref}

\end{document}